\documentclass[letterpaper, 10 pt, conference]{ieeeconf}  
\IEEEoverridecommandlockouts             
\usepackage{amssymb} %

\usepackage{bm}
\usepackage{enumerate}
\usepackage{commath}
\usepackage{graphicx}
\graphicspath{{Pictures/}}
\usepackage{amsmath}
\usepackage{mathtools}
\usepackage{breqn}
\usepackage{cite}
\usepackage[table]{xcolor}
\usepackage{booktabs}
\usepackage{empheq}
\usepackage{amsfonts}

\usepackage{amsthm}
\usepackage{hyperref}
\usepackage{xcolor}
\usepackage{mathrsfs}
\usepackage{accents}
\usepackage{thmtools}
\usepackage{thm-restate}
\usepackage{float}
\usepackage{xcolor}
\usepackage[normalem]{ulem}
\usepackage{comment}

\usepackage{setspace}

\makeatletter
\let\MYcaption\@makecaption
\makeatother
\usepackage[font=footnotesize]{subcaption}
\makeatletter
\let\@makecaption\MYcaption
\makeatother

\usepackage{algorithm}
\usepackage{algpseudocode}
\usepackage{dblfloatfix}

\theoremstyle{plain}

\newtheorem{lemma}{Lemma}

\newtheorem{proposition}{Proposition}

\newtheorem*{theorem*}{Theorem}
\newtheorem{assumption*}{Assumption}

\declaretheorem[name=Theorem]{thm}
\DeclareMathOperator{\atan}{atan}

\theoremstyle{definition}

\newtheorem{remark}{Remark}
\newtheorem{definition}{Definition}

\newtheorem*{problem*}{Problem}

\newcommand{\myvar}[1]{\bm{#1}}

\newcommand{\myset}[1]{\mathcal{#1}}

\DeclareMathOperator*{\argmin}{arg\,min}

\title{ \LARGE \bf
    Zero-Order Control Barrier Functions for Sampled-Data Systems \\
    with State and Input Dependent Safety Constraints
}

\author{Xiao Tan$^{\star}$, Ersin  Da\c{s}$^{\star}$, Aaron D. Ames, and Joel W. Burdick%
\thanks{$^{\star}$Both authors contributed equally.}%
\thanks{This work was supported by TII under project \#A6847 and DARPA under
the LINC Program.}
\thanks{The authors are with the Department of Mechanical and Civil Engineering, California Institute of Technology, Pasadena, CA 91125, USA. ${\tt\small \{xiaotan, ersindas, ames, jburdick \}@caltech.edu}$ } }

\begin{document}

\maketitle

\begin{abstract}
We propose a novel zero-order control barrier function (ZOCBF) for sampled-data systems to ensure system safety. Our formulation generalizes conventional control barrier functions and straightforwardly handles safety constraints with high-relative degrees or those that explicitly depend on both system states and inputs. The proposed ZOCBF condition does not require any differentiation operation. Instead, it involves computing the difference of the ZOCBF values at two consecutive sampling instants. We propose three numerical approaches to enforce the ZOCBF condition, tailored to different problem settings and available computational resources. We demonstrate the effectiveness of our approach through a collision avoidance example and a rollover prevention example on uneven terrains. 
\end{abstract}

\section{Introduction}

Controlling dynamical systems subject to both input and state constraints is a common challenge across engineering applications \cite{huang2021stability}. Control barrier functions (CBF) \cite{Ames2017} are a popular means to address constrained control problems. CBFs have been applied to numerous applications, including adaptive cruise control \cite{Ames2014}, bipedal robotic walking \cite{hsu2015control}, safe docking control \cite{saradagi2022safe}, safe coordination of multiple robots \cite{fernandez2023distributed}, among many others. The success of this approach is based on several appealing features.  Theoretically, CBFs provide sufficient and necessary Lyapunov-like conditions for safety, guaranteeing forward invariance and asymptotic convergence to a safe set. Practically, for a large class of nonlinear continuous-time systems and state dependent safety constraints, the CBF condition forms a linear constraint on the system's input, making it both flexible and easy to implement---often via a \emph{safety filter} \cite{wabersich2023data}.

However, there exist limitations to the conventional CBF approach. Safety constraints must be independent of the system input, and the function defining the safe set must have relative degree $1$. Otherwise, the CBF condition is either not well-defined or becomes non-binding. In order to accommodate these cases, several extensions to conventional CBFs have been proposed.  %
Notably, exponential CBFs \cite{nguyen2016exponential} and high-order CBFs~\cite{xiao2019control,tan2021high} are proposed for addressing the high-relative degree issue. However, checking whether the recursively defined safety function is a valid CBF in this approach is difficult in general, and the performance is highly sensitive to the chosen coefficients. Alternatively, back-stepping CBFs \cite{taylor2022safe} define barrier functions based on virtual inputs and tracking errors. However, back-stepping can encounter an ``explosion of terms'' problem when the relative degree is large. On the other hand, when the safety constraints depend both on system states and system inputs, integral CBFs~\cite{ames2020integral} and control-dependent CBFs~\cite{huang2019guaranteed} extend the system dynamics via a new virtual input, treat the original state and input as the new state, and then apply the conventional CBF approach for this extended system. For all these extensions, the relation between the input and its effects on the safety constraints is not explicit.

Digital implementation of a controller designed for continuous-time systems usually involves a zero-order hold, resulting in a sampled-data system formulation: the state evolves in continuous-time, while the control input is updated at discrete sampling instants, and held constant between those instants. Many works have investigated how to reason about safety for the sampled-data system with CBFs, using a safe-triggering strategy \cite{yang2019self} or robustifying the conventional CBF condition \cite{cortez2019control,singletary2020control,breeden2021control, roque2022corridor}. 
On the other end of the spectrum, discrete-time CBFs were proposed in \cite{agrawal2017discrete} for addressing safety-critical control of discrete-time dynamical systems. The discrete-time CBF condition compares the difference in the CBF values at every state transition. Inspired by this formulation, directly comparing the CBF values at sampling instants was also applied to sampled-data systems in \cite{breeden2021control,taylor2022safety}. However, a direct application of a discrete-time CBF condition does not provide a theoretic safety guarantee for trajectories during the period between the two sampling instants. Moreover, as will be clarified later in the paper, if the approximate discrete-time model is not well-chosen for a continuous-time system, then a high-relative degree constraint for that system will remain a high-relative degree for the discrete-time system. This phenomenon would require the discrete-time high-order CBF technique \cite{xiong2022discrete} to derive a valid CBF condition, further complicating the safety-critical control design.

In this work, we consider the safety-critical control problem for continuous-time control-affine systems with state and input dependent safety constraints, where the control decision is computed at discrete sampling instants. We propose a zero-order control barrier function (ZOCBF), which does not require a relative degree assumption on the constraint. Instead of requiring a condition on the system input $u$ based on the instantaneous system state, the ZOCBF condition is based on the predicted value of the safety constraint function at the next sampling step, taking a form similar to that in the discrete-time CBF condition. We further show that our ZOCBF approach also guarantees safety over the inter-sampling period. We then present three numerical approaches for enforcing the proposed ZOCBF condition and show its effectiveness for safety constraints with high-relative degrees or that are state-input dependent.

The remainder of the paper is organized as follows. Section II introduces preliminaries and our problem formulation. Section~\ref{sec:zero-order CBF} proposes a novel \emph{zero-order CBF} notion, which addresses general nonlinear state-input constraints. We also demonstrate its applicability in enforcing state constraints with high-relative degrees. In Section~\ref{sec:implementation}, we present three implementation approaches, ranging from model-based linearization to numerical integration to parallel simulation.

\textbf{Notation:} Let $\mathbb{R}$ and $\mathbb{R}_{+}$ denote  the reals and non-negative reals. A continuous-time signal $\myvar{x}: \mathbb{R} \to \mathbb{R}^n$ is denoted with boldface. A continuous function $\gamma: \mathbb{R} \to \mathbb{R}$ is an extended class $\mathcal{K}$ function if it is strictly increasing and $\gamma(0) = 0$.

\section{Preliminaries and problem formulation}

Consider continuous-time control-affine nonlinear systems
\begin{equation} \label{eq:system}
    \dot{\myvar{x}} = f(\myvar{x}) + g(\myvar{x})\myvar{u},
\end{equation}
where the state $\myvar{x}: \mathbb{R}_{+} \to \mathbb{R}^n$ is a continuous-time signal evolving in $\mathbb{R}^n$, the time-varying input $\myvar{u}:\mathbb{R}_{+} \to U$ evolves in the control set $ U \subset \mathbb{R}^m$, and the vector fields $f, g$ are locally Lipschitz continuous. Digital implementation of any controller introduces a sample-and-hold effect on the control
\begin{equation}
\label{eq:ZOH}
\myvar{u}(t) = u_k, \quad \forall t \in [t_k, t_{k+1}),
\end{equation}
where $u_k$ is constant between consecutive sampling instants $t_k$ and $ t_{k+1}$. We assume a uniform sampling strategy with a sampling interval $T$. As the system is piece-wise locally Lipschitz continuous, it admits a unique trajectory. Denote by $\phi(t; x, \myvar{u})$ the point reached by the trajectory starting at $x$ under input signal $\myvar{u}$ at time $t$. When the input signal is constant, i.e., $\myvar{u}(\tau) = v$ for $\tau \in [0,t)$, we also denote it by $\phi(t; x, v) $.  

In many applications, a system is expected to operate under a constraint that explicitly involves both system states and inputs. Suppose that this safety constraint is characterized by the zero-superlevel set of a function ${h \!:\! \mathbb{R}^n \!\times\! U \!\to\! \mathbb{R}}$, namely,
\begin{equation} \label{eq:safety_constraint}
    h(\myvar{x}(t),\myvar{u}(t)) \geq 0, \forall t \in \mathbb{R}_{+}.
\end{equation}
Denote the corresponding safe set by $\myset{C}\triangleq\{(x,u)\in \mathbb{R}^n \times U: h(x, u) \geq 0\}$. One such example is a zero-tilting moment point (ZMP)-based rollover avoidance constraint for mobile robots \cite{sardain2004forces}, which explicitly depends upon position, velocity, and torque commands. 

To accommodate the state-input dependent safety constraint \eqref{eq:safety_constraint} using a CBF, common approaches include a state projection-based approach and an integral CBF approach~\cite{ames2020integral}. The former approach constructs a CBF for the projected safe state set $  \myset{C}_S = \{x \in \mathbb{R}^n: \exists u\in U \textup{ s.t. } h(x, u)\geq 0\}$, derives relevant CBF conditions, and imposes the constraint $ h(x, u)\geq 0$ as a nonlinear constraint, together with the CBF condition. When these conditions are met by a control input, the feedback renders the projected state set $\myset{C}_S$ forward invariant. The latter approach computes a virtual input for a dynamically extended system. In general, the former approach involves constructing a CBF, solving a nonlinear constrained optimization problem, and potentially encountering constraint incompatibilities. The integral CBF approach does not provide a clear relation between the actual input and its effects on the safe constraint function value.

In this paper, we aim to develop a new CBF methodology for the sampled-data system in \eqref{eq:system}-\eqref{eq:ZOH} that handles both high-order constraints as well as state-input dependent constraints in a straightforward manner and provides safety guarantees at all times.

\section{Zero-order control barrier functions} \label{sec:zero-order CBF}

This section proposes a \emph{zero-order CBF} notion for sampled-data systems. Although the proposed CBF framework shares similarities with the discrete-time CBF formulation \cite{agrawal2017discrete} for discrete-time systems, our method provides safety guarantees for general state-input dependent constraints at all times. We further show that our method does not require relative degree information about the constraints.

\begin{definition}[Zero-Order Control Barrier Functions]
\label{def:dCBFs}
A continuous safety constraint function $h$ is a \textit{zero-order control barrier function} (ZOCBF) for the sampled-data system \eqref{eq:system}-\eqref{eq:ZOH} with a sampling time interval $T$, if for any ${(x_0, u_0)   \in   \mathbb{R}^n   \times U}$, there exists an input $u\in U$ such that:  
\begin{equation} \label{eq:zocbf}
 h(\phi(T;x_0,u), u) - h(x_0,u_0) \geq - \gamma(h(x_0,u_0)) + \delta,
\end{equation}
where $\delta$ is a positive constant to be specified later,  $\gamma(\cdot) $ is an extended class $\mathcal{K}$ function and $| \gamma(s)|\leq |s|$ for $s\in \mathbb{R}$. Denote the set of inputs that satisfy the ZOCBF condition by $U_{zocbf}(x_0, u_0) = \{u\in U: \eqref{eq:zocbf} \textup{ holds}\}$.

\end{definition}

When interpreting the ZOCBF condition at sampling time $t_k$, the state-input pair $(x_0,u_0)$ in Definition~\ref{def:dCBFs} represents the state at $t_k$ and the input from an infinitesimal moment ago, i.e., the constant input $u_{k-1}$ applied during the time interval $[t_{k-1}, t_k)$. Thanks to a continuity property, $\lim_{t\to t_{k}^{-}} h(\myvar{x}(t), \myvar{u}(t)) = h(x_k, u_{k-1})$. This ZOCBF condition limits the change in the values of $h$ at consecutive sampling instants $t_{k}$ and $t_{k+1}$. The parameter $\delta$ serves as a robustness buffer to account for the sampling effects, as shown below. We note that the ZOCBF condition in \eqref{eq:zocbf} can be restricted to the set $\myset{C}$ when only the forward invariance of the safe set is concerned. 

\begin{lemma} \label{lem:delta}
    Assume that $\frac{\partial h}{\partial x}(x,u)$ and $f(x) + g(x)u$  are bounded for any $ (x, u)\in \myset{C}$. Specifically, there exist constants ${\bar{h}_x >0}$ and ${M > 0 }$ such that ${\| \frac{\partial h}{\partial x} (x,u)\| \leq \bar{h}_x}$ and ${\| f(x) + g(x)u\| \leq M}$. If $\delta$ satisfies 
    \begin{equation} \label{eq:delta}
        \delta\geq \bar{h}_x M T,
    \end{equation}
    then the following implication holds:
    \begin{multline*}
       h(\phi(T;x,u), u) \geq \delta \\ \implies h(\phi(t;x,u), u) \geq 0 \textup{ for all  } t\in [0,T].
    \end{multline*}
\end{lemma}
\begin{proof}
Observe that, for any $ t\in [0,T]$, 
\begin{equation*}
\begin{aligned}
        | h (\phi&(T;x,u), u) - h(\phi(t;x,u), u) | \\   & \leq \left| \int_{t}^{T} \frac{\partial h}{\partial \phi}(f(x(\tau)) + g(x(\tau)) u(\tau)) \ d\tau \right| \\
        & \leq \int_{t}^{T} \left\|  \frac{\partial h}{\partial \phi} \right\| \| f(x(\tau)) + g(x(\tau)) u(\tau)\| \ d\tau \\
        & \leq \bar{h}_x M (T-t) \leq  \bar{h}_x M T.
\end{aligned}
\end{equation*}
 Thus, when $ h(\phi(T;x,u), u)\geq \delta$,  we obtain that  $$
 h (\phi(t;x,u), u) \geq \delta - \bar{h}_x M T \geq 0, \  \forall t\in [0,T].  \eqno\qedhere
$$
 \end{proof}
We note that because the sampling period $T$ is small in practice, $\delta$ in \eqref{eq:delta} is typically small in magnitude.

\begin{thm}
\label{teo:zocbfdef}
    Consider the sampled-data system \eqref{eq:system}-\eqref{eq:ZOH} with the sampling period $T$. If $h: \mathbb{R}^n   \times U   \to \mathbb{R}$ is a \textit{zero-order control barrier function} (ZOCBF) with $\delta$ satisfying \eqref{eq:delta}, then any sampled-and-hold controller $\myvar{u}(t) = u_k, t\in [t_k, t_{k+1})$ that satisfies the ZOCBF condition \eqref{eq:zocbf}, i.e., $u_k\in U_{zocbf}(\myvar{x}(t_k),\myvar{u}(t_{k-1}))$, which is given in Definition~\ref{def:dCBFs}, renders the system safe at all times. That is, the system trajectory satisfies the following implications:
    \begin{equation} \label{eq:safety guarantee}
        \begin{aligned}
        \exists v\in U, \ h(x_0,v) \geq 0 &  \implies h(\myvar{x}(t), \myvar{u}(t)) \geq 0, \forall t \geq 0, \\
        \forall v\in U, \    h(x_0,v) < 0 & \implies \lim_{t\to \infty}  \| (\myvar{x}(t), \myvar{u}(t)) \|_{\myset{C}} = 0,
        \end{aligned}
    \end{equation}
    where $x_0 $ is the initial state of the system.
\end{thm}

\begin{proof}
    We show the first implication in \eqref{eq:safety guarantee} using recursion. First, we virtually extend the trajectory and input signals backward in time by a duration $T$ by defining $\myvar{x}(t) = \phi(t;x_0,v), \myvar{u}(t) = v$ for $t\in [-T,0)$. Such a backward-time flow exists because by reversing the time variable, the terminal condition problem becomes an initial condition problem, and noting that the vector fields $f, g$ are locally Lipschitz. Note that no safety guarantee is asserted for this period. During time interval $t\in[0,T)$, since $h$ is a ZOCBF, a safe input $u_0$ exists and satisfies the ZOCBF condition, which leads to 
    \begin{equation*}
    \begin{aligned}
       h(\myvar{x}(T), u_0) & \geq (\text{Id} - \gamma)\circ h(\myvar{x}(0),v) + \delta\\
       & = (\text{Id} - \gamma)\circ h(x_0,v) + \delta \geq \delta,
    \end{aligned}
    \end{equation*}
    where $\text{Id}(\cdot)$ is the identity map.  The second inequality holds because $s - \gamma(s)\geq 0$ for any $s\geq 0$ due to the property of $\gamma$, and $h(x_0, v)\geq 0$ is one of the premises. From Lemma~\ref{lem:delta}, since $\delta$ satisfies \eqref{eq:delta}, we know that $h(\myvar{x}(t),\myvar{u}(t))\geq 0$ for $t\in [0,T)$. Now we assume that $h(\myvar{x}(t),\myvar{u}(t))\geq 0$ for $t\in [t_{k-1},t_k)$. An input $u_k$ can be found at the sampling time $t_k$ that satisfies the ZOCBF condition. We can similarly show that $h(\myvar{x}(t),\myvar{u}(t))\geq 0$ for $t\in [t_{k},t_{k+1})$. By recursion, we conclude that $h(\myvar{x}(t), \myvar{u}(t)) \geq 0$ for all $ t \geq 0$.

  To prove the second implication in \eqref{eq:safety guarantee}, we consider a similar backward trajectory extension for the interval $ [-T, 0)$ with any $v\in U$. At a sampling time $t_k\in \{0, T, \ldots \}$, if $h(\myvar{x}(t_k), \myvar{u}(t_{k-1}))<0$, from the ZOCBF condition, a suitable $u_{k}$ then exists such that $h(\myvar{x}(t_{k+1}), \myvar{u}(t_{k})) - h(\myvar{x}(t_k), \myvar{u}(t_{k-1})) \geq -\gamma(h(\myvar{x}(t_k)), \myvar{u}(t_{k-1}))+ \delta > \delta $. That is, the value of ZOCBF function $h(\myvar{x}(t_{k}), \myvar{u}(t_{k-1}))$ will increase by at least a constant $\delta$ at two consecutive sampling instants. Thus, after a finite number of steps, the value of the ZOCBF function at sampling times will become positive. Following a similar argument as above, we know the ZOCBF value will become always non-negative after a finite time. Thus, the second implication holds.
\end{proof}

Will the control input $u$ always appear in ZOCBF condition \eqref{eq:zocbf}? In the conventional continuous-time CBF literature~\cite{Ames2017}, when the constraint has a high-relative degree, i.e., when $L_g h(x) = 0$ for all $x\in \mathbb{R}^n$, the input will not appear in the CBF condition, rendering it non-binding. Particular techniques are proposed to address this issue, including exponential\cite{nguyen2016exponential}, high-order\cite{xiao2019control,tan2021high} and backstepping CBFs \cite{taylor2022safe}. For the proposed ZOCBF condition \eqref{eq:zocbf}, when the safety constraint $h$ explicitly depends upon the control input, it will naturally appear. Consider the case when the constraint function $h$ does not explicitly depend on the system input but imposes a high-relative-degree constraint. Below, we show that the input appears in our ZOCBF condition. 

Suppose for now that the safety constraint function $h$ depends only on the system state $x$ and has relative degree $\rho>1$. By taking $h(x,\cdot)$ as the system output, $y(x)= h(x,\cdot)$, we can (locally) construct a diffeomorphism $\mathcal{T}:\mathbb{R}^n \to \mathbb{R}^n$ that transforms system \eqref{eq:system} into its normal form \cite[Theorem 13.1]{Khalil2002}. Let the transformed coordinates be ${ z \!=\!  \mathcal{T}(x) \!:= \! (\xi, \eta) }$, where $\xi \in \mathbb{R}^{\rho}$ represents the external state, and $\eta\in \mathbb{R}^{n-\rho}$ the internal state. The transformed dynamics become:
\begin{equation} \label{eq:transformed_dyn}
    \begin{aligned}
        & \dot{\xi}_1  = \xi_2, ~\dot{\xi}_2  = \xi_3, \ldots, \dot{\xi}_{\rho-1}  = \xi_{\rho}, \\
        & \dot{\xi}_{\rho}  = r(x)u + a(x), \\
        & \dot{\eta} = f_0(\eta,\xi), \\ 
        & y(x)  = [1, 0, \ldots, 0] \xi = \xi_1,
    \end{aligned}
\end{equation}
where $r(x) := L_g L_f^{\rho -1} y(x)\neq 0, a(x) :=L_f^{\rho} y(x) $. More details can be found in \cite[Chapter 13.2]{Khalil2002}. Now consider the ZOCBF condition \eqref{eq:zocbf} at a sampling time $t_k$. Since the dynamics of $\xi$ form a chain of integrators and $r(x)\neq 0$ (due to $h$ having relative degree $\rho$), input $u$ will appear in $\xi_1(t_{k+1})$, or equivalently, in $h(\myvar{x}(t_{k+1}),\cdot)$. We thus conclude that $u$ will appear in the ZOCBF condition \eqref{eq:zocbf}.

\begin{remark}  
     We use the term {\em zero-order CBF} for two reasons: 1)~It is designed for sampled-data systems where the input signal is held constant using a zero-order hold; 2)~In contrast to conventional CBFs or high-order CBFs \cite{xiao2019control,tan2021high}, which require taking (high-order) time derivatives of the safety constraint to derive the CBF condition, the ZOCBF condition \eqref{eq:zocbf} requires zero differentiation operations. 
\end{remark}

\begin{remark}[ZOCBFs, discrete-time CBFs, and conventional CBFs] \label{rem:cbf_generalization}
    When the safety constraint function $h$ depends only on the system state $x$, we can define another function $b(x) := h(x,u)$. If the sampling time $T$ is small such that $\delta$ in \eqref{eq:zocbf} is negligible, then the ZOCBF condition becomes:
    \begin{equation} \label{eq:simplied_condition}
        b(\phi(T; x, u)) - b(x) \geq - \gamma(b(x)).
    \end{equation}
    \emph{Discrete-time CBFs:} The simplified condition \eqref{eq:simplied_condition} seems exactly the same as the discrete-time CBF condition in \cite{agrawal2017discrete}. Compared to a discrete-time CBF, by adding the robustness term $\delta$, the ZOCBF condition \eqref{eq:zocbf} guarantees system safety throughout the sampling intervals, while the discrete-time CBF condition only guarantees safety at the sampling time instants. From Lemma~\ref{lem:delta}, $\delta$ is small if the sampling rate is fast. It is important to note that, for some discrete-time systems and constraints, the constraints may have a high-relative degree in the discrete-time formulation.  In these cases, one must resort to applying high-order discrete-time CBFs \cite{xiong2022discrete}, which requires the control system to advance/backtrack by multiple sampling steps before the system input appears. In strong contrast, we have shown that the system input will always appear in our proposed ZOCBF condition, regardless of the safety constraint's relative degree. For numerical implementations, this feature implies that we need not worry about high-relative degree constraints in the ZOCBF framework if we can ``accurately'' predict the ZOCBF value at the next sampling instant. This difference, as well as implementation details, will be discussed in Section~IV.B and demonstrated in a numerical example in Section~V.A.
    
     \emph{Conventional CBFs:} If we approximate the trajectory flow using a forward Euler method, i.e., $ \phi(T; x, u) \approx x + (f(x) + g(x)u)T $, and approximate the function $b$ using first-order Taylor expansion, i.e., $b(\phi(T; x, u)) \approx b(x) + \frac{\partial b}{\partial x} (\phi(T;x,u) - x)$ and substitute these terms in \eqref{eq:simplied_condition}, then we recover the conventional continuous-time CBF condition:
    \begin{equation} \label{eq:conventional CBF}
        \frac{\partial b}{\partial x} (f(x) + g(x)u) + (\gamma/T)\circ b(x)  \geq 0.
    \end{equation}
   It is known that the conventional CBF is not directly applicable for scenarios where the safety constraint function depends both on state and input, or this function has a high-relative degree. Instead of resorting to high-order CBFs or integral CBFs, our proposed ZOCBF formulation can handle both cases more straightforwardly.
\end{remark}

\begin{remark}[ZOCBF and vanilla CBF constraint]
    There are two reasons why a basic safety constraint function $h(\cdot)$ may not be suitable to be used as it is for an online optimization-based controller (e.g., a model predictive control (MPC)~\cite{camacho2013model}).  
    First, this vanilla safety constraint may be nonlinear, and thus not suitable for constraining an online optimization problem, especially when it is used in each step of a multiple-step look-ahead planner, such as MPC. On the contrary, %
    our ZOCBF condition can be implemented as a convex, or even linear, constraint--see the next section. Even if the ZOCBF is nonlinear, safety is guaranteed using only one-step look-ahead prediction. Second, when directly embedded into an online nonlinear MPC,
    the vanilla constraint starts to influence the system %
    only when the (predicted) trajectory reaches the safety boundary, thus resulting in aggressive maneuvers with small safety margins. The influence of the safety constraint goes into effect much earlier if the ZOCBF condition is instead applied, leading to smoother and safer trajectories.
\end{remark}

It is important to note that in the search for a safe (and possibly optimal) input, the ZOCBF condition~\eqref{eq:zocbf} involves computation of system flows and imposes a nonlinear constraint on input $u$. The exact relation between input $u$ and the value of the safety constraint function at the next sampling instant is usually difficult to characterize analytically. To address this issue, we provide several numerical implementations in the following section.

\section{Safe control and numerical implementations} \label{sec:implementation}

   This section demonstrates how to apply our proposed ZOCBF condition as an online safety filter. Similar to conventional CBFs, for a given nominal input signal $\myvar{u}_{\text{nom}}: \mathbb{R}_{+} \to \mathbb{R}^m$, the filtered safe input is implemented via online optimization:
\begin{equation} \label{eq:safety_filter}
    \begin{aligned}
       \myvar{u}(t) & = \argmin_u \| u - \myvar{u}_{\textup{nom}}(t_k) \| \\
       \textup{s.t. } &  u \in  U_{zocbf}( \myvar{x}(t_k),\myvar{u}(t_{k-1}))
    \end{aligned}, \textup{ for $t\in [t_k, t_{k+1})$.}
\end{equation}
This optimization yields a control input that is safe and has the least deviation from the nominal control input. The nominal input signal may come from different schemes, for example, a state feedback controller or a human operator.
    
The exact ZOCBF condition can be nonlinear and not easily expressed as an analytical expression. As shown in Remark \ref{rem:cbf_generalization}, when using a forward Euler integration method to numerically approximate the flow and first-order Taylor expansion that approximates the constraint function $h$, our ZOCBF reduces to the conventional continuous-time CBF condition in \eqref{eq:conventional CBF} that is linear in the input $u$. This approach fails to consider constraints dependent on input or having a high-relative degree. We next propose several numerical approximations of the ZOCBF condition based on various techniques, including local linearization, numerical integration, and parallel simulation. We also provide some guidelines on the choice of an implementation, depending on the system dynamics and onboard computational resources.

\subsection{Approach 1: dynamics linearization}
We now introduce a method to numerically approximate the ZOCBF condition using dynamics linearization. We show that this approach retains certain convexity properties.

Define the following notations. At sampling time $t_k$, let $x_k$ be the current state, $u, u_{k-1}$ the inputs for the durations $[t_{k}, t_{k+1}), [t_{k-1}, t_{k})$, respectively. We propose to use a linear model to approximate the dynamics around the state $x_k$:
\begin{equation} \label{eq:vf linear_model}
    \dot{\myvar{\xi}} = A \myvar{\xi} + B u + C,\  \myvar{\xi}(0) = x_k, \ t\in [0,T],
\end{equation}
where $A := \frac{\partial f}{\partial x}(x_k) , B := g(x_k), C := f(x_k) -  \frac{\partial f}{\partial x}(x_k) x_k$. Unlike most dynamics linearization settings, here $x_k$ does not need to be an equilibrium point. Integrating the continuous-time time-invariant system \eqref{eq:vf linear_model} yields the solution $    \myvar{\xi}(t) = e^{At}  x_k + \int_0^t e^{A(t - \tau)} d \tau  (Bu + C).$
Our predicted state $\hat{x}_{k+1}$ at $t_{k+1}$ is  linear in $u$:
\begin{equation} \label{eq:x_k+1 linearization}
     \hat{x}_{k+1} = \myvar{\xi}(T) = A_{D} x_k + B_{D}Bu + B_{D}C,
\end{equation}
 where $A_{D}:= e^{AT},~B_{D}: = \int_0^T e^{A(T - \tau)} d \tau$.

We can approximate the function $h(x,u)$ around the state-input pair $(x_k, u_{k-1})$ using first- or second-order Taylor expansion, i.e., 
\begin{multline}
    \label{eq:first-order Taylor}
    \tilde{h}_1(x,u) \!=\! h(x_k, u_{k-1}) +\! \begin{bsmallmatrix}
        \nabla_x h (x_k, u_{k-1}) \\
        \nabla_u h (x_k, u_{k-1})
    \end{bsmallmatrix}^\top \begin{bsmallmatrix}
        x - x_{k} \\
        u - u_{k-1}
    \end{bsmallmatrix} 
\end{multline} 
or
\begin{multline}
    \label{eq:h quadratic_appro}
    \tilde{h}_2(x,u) = h(x_k, u_{k-1}) + \begin{bsmallmatrix}
        \nabla_x h (x_k, u_{k-1}) \\
        \nabla_u h (x_k, u_{k-1})
    \end{bsmallmatrix}^\top \begin{bsmallmatrix}
        x - x_{k} \\
        u - u_{k-1}
    \end{bsmallmatrix}  \\ + \begin{bsmallmatrix}
        x - x_{k} \\
        u - u_{k-1}
    \end{bsmallmatrix}^\top \mathcal{J}_{h} (x_k, u_{k-1})  \begin{bsmallmatrix}
        x - x_{k} \\
        u - u_{k-1}
    \end{bsmallmatrix},
\end{multline} 
where $\mathcal{J}_h$ is the Jacobian of the function $h$. The approximate ZOCBF condition is obtained by substituting our predicted state $\hat{x}_{k+1}$ and the approximate safety constraint function in \eqref{eq:first-order Taylor} or \eqref{eq:h quadratic_appro} into the ZOCBF condition in~\eqref{eq:zocbf}.  

When the linear approximation of the safety constraint function in \eqref{eq:first-order Taylor} is applied, the approximate ZOCBF condition becomes:
        \begin{equation}
            \left( \nabla_x h B_D B +  \nabla_u h \right) u    +
            \kappa(x_k, u_{k-1},\delta) \geq 0 , 
        \end{equation}
where we define $\kappa(x_k, u_{k-1},\delta) := \nabla_x h (A_D x_k + B_D C - x_{k})            -  \nabla_u h u_{k-1} + \gamma(h(x_{k},u_{k-1})) - \delta $ and the partial derivatives $ \nabla_x h, \nabla_u h$ are shorthanded for  $\nabla_x h(x_k, u_{k-1}),\nabla_u h(x_k, u_{k-1}) $, respectively. In this case, the approximate ZOCBF condition is a linear constraint on $u$. We can similarly derive the approximate ZOCBF condition using the quadratic approximation of $h$ in \eqref{eq:h quadratic_appro} as
    \begin{multline} \label{eq:ZOCBF linear_quadratic approx}
    \begin{bsmallmatrix}
        A_{D} x_k + B_{D}u + B_{D}C - x_{k} \\
        u - u_{k-1}
    \end{bsmallmatrix}^\top \mathcal{J}_{h}  \begin{bsmallmatrix}
        A_{D} x_k + B_{D}u + B_{D}C - x_{k} \\
        u - u_{k-1}
    \end{bsmallmatrix} \\ 
    + \left( \nabla_x h B_D +  \nabla_u h \right) u + \kappa(x_k, u_{k-1},\delta) \geq 0,
\end{multline} 
 where $\kappa(x_k, u_{k-1},\delta) $ is defined as above and $ \mathcal{J}_{h}$ is shorthanded for $\mathcal{J}_{h}(x_{k}, u_{k-1})$. We establish the following result:

\begin{proposition}
If function $h$ is concave at $(x_k, u_{k-1})$, then the approximate ZOCBF condition  \eqref{eq:ZOCBF linear_quadratic approx} defines a convex quadratic constraint set on $u$.
\end{proposition}

\begin{proof}
Since the function $h$ is concave at $(x_k, u_{k-1})$, $ \mathcal{J}_h(x_k, u_{k-1})$ is a negative semidefinite matrix. We derive the quadratic term in $u$ in the left-hand side of \eqref{eq:ZOCBF linear_quadratic approx} as$$u^\top \begin{bsmallmatrix}
    B_D \\
    I
\end{bsmallmatrix}^\top \mathcal{J}_h(x_k, u_{k-1}) \begin{bsmallmatrix}
    B_D \\
    I
\end{bsmallmatrix} u.$$
Thus, the quadratic coefficient matrix is negative semidefinite. The remaining term in the left-handside of \eqref{eq:ZOCBF linear_quadratic approx} is linear in $u$. Therefore, the approximate ZOCBF in \eqref{eq:ZOCBF linear_quadratic approx} defines a convex 
 quadratic set on $u$.
\end{proof}

\subsection{Approach 2: numerical integration}

We can alternatively use numerical integration methods (forward Euler, Runge-Kutta (RK), etc) to approximate the flow map. One can generally tune the desired precision of the approximate flow at the expense of computation burden. Compared to dynamic linearization, numerical integration does not use differentiation of the nonlinear vector fields but instead uses function evaluations. This approach usually yields an explicit yet nonlinear relation between the input $u_k$ in the time interval $[t_{k}, t_{k+1})$ and the predicted state $\hat{x}_{k+1}$ of the safety constraint at $t_{k+1}$. In the following, we show that by substituting the predicted state $\hat{x}_{k+1}$ in the ZOCBF condition~\eqref{eq:zocbf}, the input will appear. For notational simplicity, we ignore the small discrepancy between the predicted and the exact states in this subsection, which will be accounted for in a following subsection.

For state-input dependent CBFs, $h(x, u)$, the input $u$ appears in~\eqref{eq:zocbf} naturally.
We restrict our discussion to state-dependent constraints $h(x,\cdot)$ here.
To explore the relationship between the order of numerical integration methods and the appearance of $u$ in the ZOCBF condition, we demonstrate the results for a widely used RK integration method. As for the continuous-time systems, we consider the transformed dynamics in its normal form as in \eqref{eq:transformed_dyn} to facilitate our discussion. The $p$-th order RK integration \cite{taylor2022safety} for an integration duration $T$ is given in the general form:
\begin{align}
\label{eq:RKn}
&z(t_{k+1}) = z(t_{k}) + T \sum_{i=1}^{p}  \kappa_i(f_z(\varrho_i) + g_z(\varrho_i)u), \\
&\varrho_i = z(t_{k}) + T \sum_{j=1}^{i-1} \lambda_{i,j}(f_z(\varrho_j) + g_z(\varrho_j)u) , \label{eq:RK_intermediate} 
\end{align}
where recall that ${z \!=\! (\xi_1,\ldots,\xi_{\rho}, \eta_1, \ldots, \eta_{n - \rho}) \!=\! \mathcal{T}(x)}\in \mathbb{R}^n$ is the transformed coordinates. We use $f_z, g_z$ for the transformed continuous-time dynamics, with their explicit expressions given below. Here intermediate states  ${\varrho_1 \!=\! z(t_{k})}$ and  ${\varrho_i, i \!=\! 2, \ldots, p}$ are recursively defined in \eqref{eq:RK_intermediate}. The parameters ${\kappa_i, \lambda_{i,j} }$ are positive constants with ${\sum_{i=1}^{p} \kappa_i \!=\! 1}$.
\begin{align*}
\begin{split}
&f_z (z) =  \begin{bmatrix}
\xi_2 &
\xi_3 &
\ldots &
\xi_{\rho} &
 a\big(\mathcal{T}^{-1}(z)\big) &
f_0(z)
\end{bmatrix}^\top, \\
&g_z(z) = \begin{bmatrix}
0 &
0 &
\ldots &
0 &
r\big(\mathcal{T}^{-1}(z)\big) &
0
\end{bmatrix}^\top.
\end{split}
\end{align*}
In view of the first row of \eqref{eq:RKn}, we obtain:
\begin{equation*}
\begin{aligned}
    z_1(t_{k+1}) & = z_1(t_{k}) + T \sum_{i=1}^{p} \kappa_i [f_z(\varrho_i) + g_z(\varrho_i) u ]_1 \\
    & = z_1(t_{k}) + T \sum_{i=1}^{p} \kappa_i [\varrho_i]_2,
\end{aligned}
\end{equation*}
Here $[\cdot]_j: (x_1, \ldots, x_n) \mapsto x_j$ is an operation that outputs the $j$th element from a vector. We apply the transformed dynamics to obtain the second equality. 
From \eqref{eq:RK_intermediate}, we know $    [\varrho_i]_2 = z_2(t_{k}) + T \sum_{j=1}^{i-1} \lambda_{i,j} [\varrho_j]_3.$
Continuing recursively for $\sigma= 3, \ldots, \rho - 1$ yields:
\begin{equation*}
[\varrho_i]_{\sigma} = z_{\sigma}(t_{k}) + T \sum_{j=1}^{i-1} \lambda_{i,j} [\varrho_j]_{\sigma+1}. 
\end{equation*}
At $[\varrho_i]_{\rho} $, where recall $\rho$ is the relative degree of the constraint, the control input $ u $ first appears:
\begin{equation*}
\label{eq:xi_rho_minus1}
[\varrho_i]_{\rho} \!=\! z_{\rho}(t_{k}) +  
T \sum_{j=1}^{i-1} \lambda_{i,j} \big(\! a\big(\mathcal{T}^{-1}(\varrho_j)\big) \!+ r\big(\mathcal{T}^{-1}(\varrho_j)\big) u \big).
\end{equation*}
Thus, $u$ appears in $z_1(t_{k+1}) $ for a constraint with relative degree $\rho$ only when we apply a $p$th-order RK integration with ${ p \!\geq \! \rho }$. In this case, the control input appears explicitly in $h(\myvar{x}(t_{k+1}))$ and influences its value, or equivalently, the ZOCBF condition.

\subsection{Approach 3: parallel simulation}
For some applications, the system dynamics and the safety constraints may be too complex to analyze analytically. Instead of using an explicit dynamical equation, as in \eqref{eq:system}, we may have access to a simulator that can compute the flow map and evaluate the function $h$ in real time. Gazebo, MuJoCo \cite{todorov2012}, and Isaac Gym \cite{makoviychuk2021}, and Brax \cite{freeman2021} simulation environments use physics engine and robot URDF description files to parallelize simulations of robotic movements. In other cases, we may have access to a physics-informed neural network (PINN) \cite{raissi2019physics} that is well-trained from input-output data, which can be evaluated quickly. In these scenarios, one can sample the control set $U$, run a large number of parallel simulations using the sampled inputs in batch, and find an input that satisfies the ZOCBF constraint in~\eqref{eq:zocbf} and minimizes the cost function in \eqref{eq:safety_filter}. Therefore, finite-sample simulations can provide exact verification of safety constraints, with high confidence, depending on the sampling density \cite{vincent2024}.

\subsection{Discussions}
All numerical implementations will create discrepancies between the approximate and the exact ZOCBF conditions. Assume that the discrepancy is bounded by function $m$:
\begin{equation}
    | h(\phi(T; x_k, u), u) - \hat{h}(\hat{x}_{k+1}, u)  | \leq m(T,u),
\end{equation}
 then, we can derive an implementable and sufficient safe condition at sampling time $t_k$ as:
\begin{equation}
   \hat{h}(\hat{x}_{k+1}, u)  \geq  (\text{Id} - \gamma)\circ (h(x_{k},u_{k-1})) + \delta + m(T,u).
\end{equation}
This condition sufficiently implies the ZOCBF condition in \eqref{eq:zocbf}, thus any feedback controller enforcing this condition guarantees system safety. 
We note that $m(T,u)$ is problem- and approach-specific. For example, if the constraint function $h:(x,u)\mapsto h(x,u)$ is Lipschitz continuous in $x$ with Lipschitz constant $c_0$ uniformly in $u$, and the fourth-order RK integration method is applied to approximate the flow, then there exists a constant $K>0$ such that  $ m(T,u) = Kc_0 T^5$. When the sampling time $T$ is small and fixed, $m(T,u)$ is a fairly small constant.

\begin{remark}
    Different approaches are suitable for different application scenarios. Dynamic linearization is most useful when the system dynamics is nearly linear. This approach provides a linear or convex quadratic constraint (supposing that the constraint function is concave) on the input. It also requires the least computation.  Numerical integration is suitable for general nonlinear dynamics models, and we can tune the precision of the predicted state via the choice of integration methods. We have shown that, 
    for high-relative degree constraints, high-order numerical integration methods yield a constraint in which the input $u$ appears. In practice, we can derive this constraint analytically using a symbolic toolbox, and include it into a nonlinear solver. 
    The sampling-based parallel simulation approach can be used when the system dynamics are complex or unknown. It however requires the most computational resources. Implementation details and comparisons of different techniques are discussed next.  
\end{remark}

\section{Simulation}
To showcase how effectively our proposed methods ensure the safety of sampled-data systems, we present two numerical examples: a double integrator system with position constraints (which have a relative degree of two), and a ground vehicle with a state-input dependent rollover constraint.\footnote{\label{footnote:code} The code and additional details of our simulations are available at \url{https://github.com/ersindas/Zero-order-CBFs}}

\subsection{Constraints with high-relative degrees}

\begin{figure}[t]
\centering
\includegraphics[width=1\linewidth]{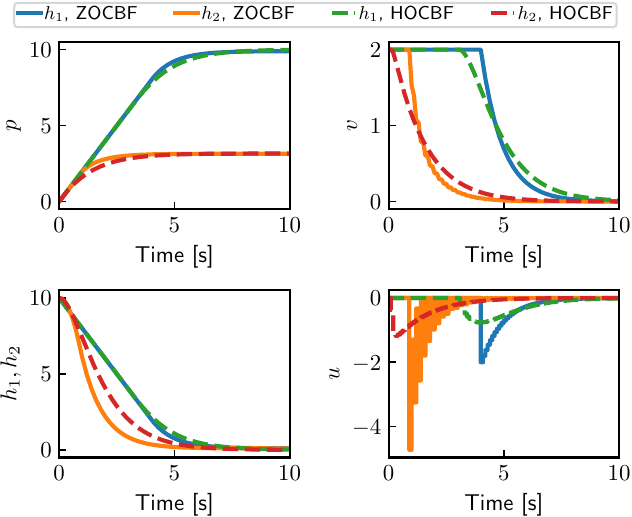}
    \caption{ Numerical simulation of the double integrator system \eqref{eq:integrator} with ${T \!=\! 0.1~s}$, ${\gamma_c \!=\! 1}$, ${u_{nom} \!=\! 0}$, ${\delta \!=\! 0.01}$, and the initial state ${x_0 \!=\! [0~2]^\top}$. The positions and velocities of the double integrator (Top). The evolution of the ZOCBFs $h_1$ and $h_2$ along the system’s trajectories and the control input over time (Bottom). In each plot, the dashed lines represent the results with discrete-time higher-order CBFs (HOCBFs) for comparison. All controllers effectively ensure safety. }
   \label{fig:integrator}
\end{figure}

Consider a double integrator model: $\dot{p} = v, \dot{v} =  u, p\in \mathbb{R}, u\in U = [-10,10]$.  Denote the state by $x = [p~v]^\top$. one can verify that the safety constraint $h_1(x) = 10 - p$ is a high-order constraint. Here we apply the first numerical approach since the system is linear. For a sampling period $T$, the discrete-time transition dynamics are 
    \begin{equation}
    \label{eq:integrator}
        \begin{bmatrix}
            p_{k+1} \\
            v_{k+1}
        \end{bmatrix} = \begin{bmatrix}
            1 & T \\
            0 & 1
        \end{bmatrix} \begin{bmatrix}
            p_{k} \\
            v_{k}
        \end{bmatrix} + \begin{bmatrix}
            0.5 T^2 \\
            T
        \end{bmatrix} u_{k}.
    \end{equation}
    Choosing the extended class $\mathcal{K}$ function $\gamma(s) = \gamma_c s$ for $s\in\mathbb{R}$ with $0<\gamma_c <1 $, we derive the ZOCBF condition in~\eqref{eq:zocbf} on the sampled input $u_k$  at sampling time $t_k$ 
    \begin{equation*}
        - 0.5 T^2 u_k  - T v_k + \gamma_c (10 - p_k) \geq \delta.
    \end{equation*}
    Thus, we obtain a ZOCBF constraint that is linear in $u_k$. Now consider another safety constraint $h_2(x) = 10 - p^2$. Again, $h_2$ is a high-order constraint. Under the same setting as above, we derive the ZOCBF condition for constraint $h_2$ on the sampled input $u_k$ at the sampling time $t_k$ 
    \begin{equation*}
       - (0.5 T^2 u_k + p_k + T v_k)^2 + p_k^2 +  \gamma_c (10 - p_k^2) \geq \delta, %
    \end{equation*}
    which is a convex quadratic constraint in $u_k$. The initial state for the simulated data shown in Fig. ~\ref{fig:integrator} is chosen to yield a safety violation under ${u_{nom}} \equiv 0$.  The figure demonstrates that the ZOCBF-based safety filter keeps the system safe.

If the double integrator dynamics are discretized in time using forward Euler integration, which is the case where the order of integration method is smaller than the relative degree, the constraint functions $h_1(x), h_2(x)$ both have a relative degree of two for the approximate discrete-time model. Conventionally, this would require the use of a discrete-time higher-order CBF technique. We also report the performance of this approach in Fig.~\ref{fig:integrator}.

In practice, one may tune the system behavior via the extended class $\mathcal{K}$ function $\gamma_c$ so that the system trajectory approaches the safety boundary more slowly or the safety filter intervenes earlier, necessitating less aggressive control inputs. In general, as discussed in Remark \ref{rem:cbf_generalization}, the term ${\gamma_c/T }$ is analogous to the extended class $\mathcal{K}$  function in conventional CBFs. Additionally, increasing the parameter $\delta$ adjusts the safety margin to prevent the system from moving too close to the boundary of the unsafe set.

\subsection{Constraints involving state and input}
Our second example looks at rollover prevention for differential drive robots that navigate over uneven terrains. The safety condition is based on the zero-tilting moment point (ZMP), as described in \cite{tipover}. This constraint depends on both the system’s state and input. To handle this constraint, we employ our proposed ZOCBF approach.

In order to develop a nonlinear dynamic model of the vehicle as it navigates 3D terrains, we derive the kinematic equations for 3D motion in an inertial frame $\mathcal{I}$.
We assume that the robot maintains ground contact, which implies that the robot's roll angle $\alpha$ and pitch angle $\beta$ equal the ground's slope angles and ${\alpha, \beta\in  (-\pi/2, \pi/2)}$.  Because the terrain varies, these angles are functions of the vehicle's positional states, denoted as $\alpha(x)$ and $\beta(x)$.  The vehicle heading angle with respect to the inertia frame ${\mathcal{I}}$ is denoted by ${{\theta} \!\in\! [0, 2 \pi ) }$.

\begin{figure}[t]
\centering
\includegraphics[width=0.70\linewidth]{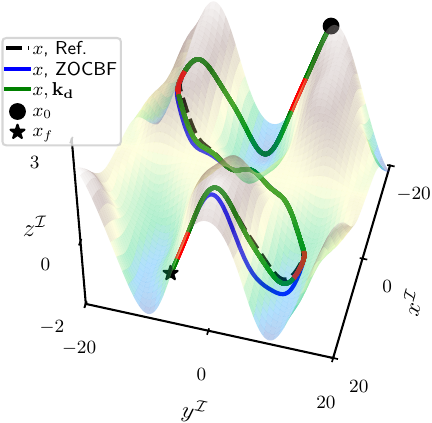}
    \caption{Simulation of a differential drive vehicle model \eqref{eq:uni_nom} navigating on uneven terrain, comparing the safety-filtered trajectory using the proposed ZOCBF filter (blue) with the unsafe nominal controller (green) while following a reference trajectory (black). The initial position is denoted as $x_0$, while the final position is denoted as $x_f$. The red areas along the blue path indicate unsafe routes on which the vehicle rollovers when under the nominal (unfiltered) controller. The vehicle is able to drive on rough terrain safely with the ZOCBF, whereas it is unsafe under the nominal controller.}
   \label{fig:3d_ZOCBF}
\end{figure}
\begin{figure}[t]
\centering
\includegraphics[width=0.8\linewidth]{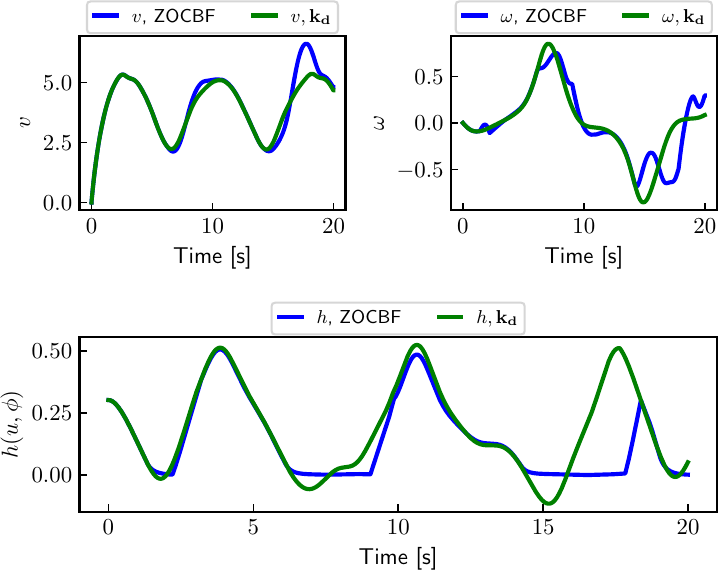}
    \caption{Control inputs with the proposed ZOCBF-based safety filter and the nominal controller over time (Top). Evolution of the rollover CBF \eqref{eq:yzmp} values for both controllers (Bottom). }
   \label{fig:h_ZOCBF}
\end{figure}

For the sake of simplicity, we assume that the robot does not slip and that the drift $f(x)$ is a zero vector. The state-space representation of the kinematic model in frame $\mathcal{I}$ is
\begin{equation}
\label{eq:uni_nom}
    \begin{bmatrix}
    \dot{x}^{\mathcal{I}} \\
    \dot{y}^{\mathcal{I}} \\
    \dot{\theta}
    \end{bmatrix}
    \!\!\!= \!\!\!
    \begin{bmatrix}
     \cos\theta \cos\beta & 0   \\
     \sin\theta \cos\beta & 0   \\
       0 &  \cos \alpha / {\cos \beta} 
    \end{bmatrix}
    \begin{bmatrix}
       v \\
      \omega  
    \end{bmatrix} ,
\end{equation}
where ${ [ {x}^{\mathcal{I}} ~ {y}^{\mathcal{I}} ]^\top \!\in\! \mathbb{R}^2}$ is vehicle position, $\omega$ represents the robot's yaw rate, and $v$ denotes linear velocity in the body frame $\mathcal{B}$. When ${\alpha = \beta = 0}$, we get the unicycle model \cite{hu2021integrated}. 

Next, we define a  path following control law based on the forward motion control law proposed in \cite{icsleyen2023} as
\begin{equation*}
    \mathbf{k_d}(x) = \begin{bmatrix}
        K_v  d_g  &
        K_{\omega} \atan \left ( d_g / {\bar{d}_g} \right )
    \end{bmatrix}^\top ,
\end{equation*}
where ${K_v,\!~K_{\omega} \!\geq\! 0}$ are the controller gains, ${[{x}_g ~y_g]^\top \!\in\! \mathbb{R}^2}$ is the goal position, and ${d_g \!\triangleq\!  \cos \theta   \left ( {x}_g \!-\! x^{\mathcal{I}} \right ) \!+\! \sin \theta   \left ( y_g \!-\!y^{\mathcal{I}} \right ) }$, ${\bar{d}_g \!\triangleq\!  -\sin \theta   \left ( x_g \!-\! x^{\mathcal{I}} \right ) \!+\! \cos \theta   \left ( y_g \!-\!y^{\mathcal{I}} \right ) }$. We set ${\omega = 0}$ when ${x_g \!=\! x^{\mathcal{I}},~y_g \!=\!y^{\mathcal{I}}}$. This acts as our nominal control input to our ZOCBF safety filter in \eqref{eq:safety_filter}.

ZMP serves as a metric to assess the risk of tipover in robotics systems. It represents the point on the ground where the combined forces of gravity and inertia generate a moment about the plane's normal direction, resulting in a tipping moment. Essentially, when the ZMP remains within the robot's support polygon—the convex hull formed by its wheel contacts—the system does not have any tendency to tipover. When the ZMP lies outside of this convex hull, the motion is considered unsafe. Rollover refers to the tipping of a robotic system around its lateral (roll) axis. The ZMP-based rollover constraint from \cite{tipover} can be expressed as 
\begin{equation}
  \label{eq:yzmp}
 h(x, u) =  -\left\vert 
    \frac{v \omega }{\mathrm{g} \cos \alpha(x) }
  \right\vert + \frac{b }{2 h_{cg}} - \tan \alpha(x) \geq 0,
\end{equation}
where $\mathrm{g}$ is the gravitational constant, $h_{cg}$  the height of the robot's center of mass, and $b$ the robot's width.  

We simulated robot operation on an undulating surface, and selected a reference path that induces rollover under the control law $\mathbf{k_d}$ (see Fig.~\ref{fig:h_ZOCBF}). We computed the system flow using a fourth-order RK numerical integration technique. We then solved the online optimization problem \eqref{eq:safety_filter} using a sequential least squares programming (SLSQP) solver. The simulation was run on a laptop with a 14-core Intel i7-12700H CPU and 32 GB of RAM, and completed in 4 seconds. The simulated closed-loop system behavior can be seen in Fig.~\ref{fig:3d_ZOCBF} and Fig.~\ref{fig:h_ZOCBF}. These figures show that the proposed method ensures the robot's rollover safety by filtering out the unsafe nominal controller through the ZOCBF-based framework. A parallel simulation-based ZOCBF safety filter is also implemented, which shows a similar performance and is neglected due to space limit. We refer interested readers to the code webpage\textsuperscript{\ref{footnote:code}} for details.

\section{Conclusion}
This paper introduced the notion of zero-order control barrier functions (ZOCBFs) for ensuring the safety of sampled-data systems. Our ZOCBF approach extends conventional CBFs by effectively and straightforwardly handling safety constraints with high-relative degrees and those dependent on both states and inputs. We presented three numerical implementation methods—dynamic linearization, numerical integration, and parallel simulation—each suited to different system dynamics and computational capabilities. Through simulations involving collision avoidance and rollover prevention on uneven terrain, we demonstrated the practical effectiveness and versatility of ZOCBFs in maintaining safety-critical operations. This advancement simplifies the design and deployment of safety filters in sampled-data systems. Future work will explore ZOCBFs implementations in higher-dimensional and more complex robotic systems.

\begin{spacing}{0.98}
\bibliographystyle{IEEEtran}
\bibliography{IEEEabrv,references}
\end{spacing}

\end{document}